\documentclass[a4paper, 10pt]{article} 
\usepackage{graphicx, amssymb} 
\usepackage[sumlimits]{amsmath}
\usepackage{amsthm}
\input{xy}
\xyoption{all}
\usepackage{array}

\usepackage[french]{todonotes}
\usepackage{algorithm}
\usepackage{algpseudocode}
\algrenewcommand{\algorithmiccomment}[1]{\hskip1em // #1}

\newtheorem{thm}{Theorem}

\newtheorem{cor}[thm]{Corollary}
\newtheorem{prp}[thm]{Proposition}
\newtheorem{lem}[thm]{Lemma}
\theoremstyle{remark}

\newtheorem{definition}[thm]{Definition}
\newtheorem{clm}[thm]{Claim}

\newcommand{\refsec}[1]{Section~\ref{sec:#1}}
\newcommand{\refalg}[1]{Algorithm~\ref{alg:#1}}
\newcommand{\reffig}[1]{Figure~\ref{fig:#1}}
\newcommand{\refeqn}[1]{(\ref{eqn:#1})}
\newcommand{\reftbl}[1]{Table~\ref{tbl:#1}}
\newcommand{\refdef}[1]{Definition~\ref{def:#1}}
\newcommand{\refthm}[1]{Theorem~\ref{thm:#1}}
\newcommand{\reflem}[1]{Lemma~\ref{lem:#1}}
\newcommand{\refprp}[1]{Proposition~\ref{prp:#1}}
\newcommand{\refcor}[1]{Corollary~\ref{cor:#1}}

\newcommand{\refline}[1]{Line~\ref{line:#1}}

\newcommand{\R}{\mathbb R}
\newcommand{\N}{\mathbb N}
\newcommand{\C}{\mathbb C}
\newcommand{\eps}{\varepsilon}
\newcommand{\cE}{{\cal E}}
\newcommand{\cD}{{\cal D}}
\newcommand{\cS}{{\cal S}}
\newcommand{\cG}{{\cal G}}

\newcommand{\cN}{{\cal N}}
\newcommand{\cP}{{\cal P}}
\newcommand{\cB}{{\cal B}}

\newcommand{\si}{\sigma_{\mathrm i}}
\newcommand{\so}{\sigma_{\mathrm o}}

\newcommand{\norm}[1]{{\lVert #1 \rVert}}
\newcommand{\braket}[2]{{\langle#1|#2\rangle}}

\newcommand{\s}[1]{\left(#1\right)}
\newcommand{\pr}{\mathop{\mathrm{Pr}}}
\newcommand{\E}{\mathop{\mathrm{E}}}

\newcommand{\ADV} {\mathrm{Adv}}

\newcommand{\pfstart}{\begin{proof}} 
\newcommand{\pfend}{\end{proof}} 

\title{Quantum Algorithm for $k$-distinctness with Prior Knowledge on the Input}
\author{Aleksandrs Belovs\thanks{Faculty of Computing, University of Latvia, Raina bulv. 19, Riga, LV-1586, Latvia, stiboh@gmail.com.} \and Troy Lee\thanks{Centre for Quantum Technologies}}
\date{}

\begin{document}
\maketitle

\begin{abstract}
It is known that the dual of the general adversary bound can be used to build quantum query algorithms with optimal complexity. Despite this result, not many quantum algorithms have been designed this way. This paper shows another example of such algorithm.

We use the learning graph technique from~\cite{spanCert} to give a quantum algorithm for $k$-distinctness problem that runs in $o(n^{3/4})$ queries, for a fixed $k$, given some prior knowledge on the structure of the input. The best known quantum algorithm for the unconditional problem uses $O(n^{k/(k+1)})$ queries.
\end{abstract}

\section{Introduction}
This paper is a sequel of~\cite{spanCert} on applications of span programs, or, more generally, dual of the Adversary Bound, for constructing quantum query algorithms for functions with 1-certificate complexity bounded by a constant. Also, we use the computational model of a learning graph. In the aforementioned paper, a reduction of a learning graph to a quantum query algorithm was done using the notion of a span program, another computational model, proven to be equivalent to quantum query algorithms in the papers of Reichardt {\em et al.}~\cite{Reichardt10advtight, LMRS11}.

Two questions remained open from the last paper. Firstly, the logarithmic increase in the complexity for functions with non-Boolean input; and whether a learning graph that uses values of the variables to weight its arcs has more power than the one that doesn't. We fully resolve the first concern by switching from span programs to a more general notion of the dual of the adversary bound that possesses the same properties, and, thus, getting a query algorithm with the same complexity as the learning graph, up to a constant factor.

For the analysis of the second problem, we have chosen the $k$-distinctness problem for $k>2$. This is the most symmetric problem for which the knowledge of the values of variables can be important in construction of the learning graph. Let us define this problem here.

The {\em element distinctness} problem consists in computing the function $f\colon [m]^n\to\{0,1\}$ that evaluates to 1 iff there is a pair of equal elements (known as {\em collision}) in the input, i.e., $f(x_1,\dots,x_n)=1$ iff $\exists i\ne j:  x_i=x_j$. The quantum query complexity of the element distinctness problem is well understood.  It is 
known to be $\Theta(n^{2/3})$, the algorithm given by Ambainis~\cite{distinct} and the lower bound 
shown by Aaronson and Shi~\cite{distinctLower1} for the case of large alphabet size 
$\Omega(n^2)$ and by Ambainis~\cite{distinctLower2} in the general case. Even more, the lower bound $\Omega(n^{2/3})$ holds if one assumes there is either no, or exactly one collision in the input.

The {\em $k$-distinctness problem} is a direct generalization of the element distinctness problem. Given the same input, function evaluates to 1 iff there is a set of $k$ input elements that are all equal. The situation with the quantum query complexity of the $k$-distinctness problem is not so clear.  As element distinctness 
reduces to the $k$-distinctness problem by repeating each element $k-1$ times, the lower bound 
of $\Omega(n^{2/3})$ carries over to the $k$-distinctness problem (this argument is attributed 
to Aaronson in \cite{distinct}).  However, the best known algorithm requires $O(n^{k/(k+1)})$ 
quantum queries~\cite{distinct}.  

As the above reduction shows, the $k$-distinctness problem can only become more difficult 
as $k$ increases.  There is another difficulty that arises when $k>2$---this the huge diversity 
in the inputs.  For element distinctness, all inputs that are distinct are essentially the same---
they are all related by an automorphism of the function.  Similarly, without loss of generality, one 
may assume that an input which is not distinct has a unique collision, and again all such inputs are 
related by an automorphism.  When $k >2$ this is no longer the case.  For 
example, for $k=3$ inputs can differ in the number of unique elements.  

\paragraph{Main theorem}
In this paper, we show how one can fight the first difficulty, but we ignore the second one. Before explaining how we do so, let us give some additional definitions.

Let $(x_i)_{i\in[n]}$ be the input variables for the $k$-distinctness problem. Assume some subset $J\subseteq [n]$ is fixed. A subset $I\subseteq J$ is called {\em $t$-subtuple} with respect to $J$ if 
\begin{equation}
\label{eqn:tuple}
\forall i,j\in I:  x_i = x_j,\qquad \forall i\in I\; \forall j\in J\setminus I:  x_i\ne x_j\qquad\mbox{and}\qquad |I|=t,
\end{equation}
i.e., if it is a maximal subset of equal elements and it has size $t$. For the important special case $J=[n]$, we call them {\em $t$-tuples}. If $I$ is such that only the first condition of~\refeqn{tuple} is satisfied, we call it {\em subset of equal elements}.

We give a quantum algorithm for the $k$-distinctness problem that runs in $o(n^{3/4})$ queries for a fixed $k$, but with the prior knowledge on the number of $t$-tuples in the input. Using the same reduction as in \cite{distinct}, it is easy to show the complexity of this problem is $\Omega(n^{2/3})$ as well.

\begin{thm}
\label{thm:main}
Assume we know the number of $t$-tuples in the input for the $k$-distinctness problem for all $t=1,\dots, k-1$ with precision $O(\sqrt[4]{n})$. Then, the problem can be solved in $O(n^{1-2^{k-2}/(2^k-1)})$ quantum queries. The constant behind the $O$ depends on $k$, but not on $n$.
\end{thm}

The precision in the formulation of the theorem can be loosened, $O(\sqrt[4]{n})$ is the most obvious value that works for all $k$'s. See \refsec{convention} for more details. Concerning the complexity of the algorithm, it is the exact one, and we do not know whether it can be improved.

\paragraph{Organization of the Paper}
The paper is organized as follows. In \refsec{prelim}, we define basic notions from quantum query complexity and probability theory. In \refsec{learning}, we define learning graphs and give a quantum algorithm for computing them. In \refsec{getting}, we develop some tools and get ready for \refsec{main}, where we prove \refthm{main}.

\section{Preliminaries}
\label{sec:prelim}
Let $[m]$ denote the set $\{1,2,\dots,m\}$ and consider a function $f\colon [m]^n\supseteq \cD \to \{0,1\}$.  We identify the set of input indices of $f$ with $[n]$. An {\em assignment} is a function $\alpha\colon [n]\supset S\to [m]$.
One should think of this function as fixing values for input variables in $S$. We say input $x=(x_i)_{i\in [n]}$ {\em agrees} with assignment $\alpha$ if $\alpha(i) = x_i$ for all $i\in S$. If $S\subseteq [n]$, by $x_S$, we denote the only assignment on $S$ that agrees with $x$. 

An assignment $\alpha$ is called a {\em $b$-certificate} for $f$ if any input from $\cD$, consistent with $\alpha$, is mapped to $b$ by $f$. The {\em certificate complexity} $C_x(f)$ of function $f$ on input $x$ is defined as the minimal size of a certificate for $f$ that agrees with $x$. The $b$-certificate complexity $C^{(b)}(f)$ is defined as $\max_{x\in f^{-1}(b)} C_x(f)$.

We use $[a,b]$ and $]a,b[$ to denote closed and open, respectively, intervals of $\R$; $\R^+$ to denote the set of non-negative reals. For the real vector space $\R^m$, we use the $\ell_\infty$-norm, $\|x \|_\infty = \max_i |x_i|$. In particular, we denote the $\ell_\infty$-ball of radius $d$ around $x$ by $\cB(x,d)$. We use $\cB(d)$ to denote the ball with center zero and radius $d$.

For the complex vector spaces $\C^m$, however, we use a more common $\ell_2$-norm, $\|x\|=\sqrt{\sum_i |x_i|^2}$.

\subsection{Adversary bound}
In this paper, we work with query complexity of quantum algorithms, i.e., we measure the complexity of a problem by the number of queries to the input the best algorithm should make. Query complexity provides a lower bound on time complexity. For many algorithms, query complexity can be analyzed easier than time complexity. For the definition of query complexity and its basic properties, a good reference is~\cite{survey}.

The adversary bound, originally introduced by Ambainis \cite{Ambainis00adversary}, is one of 
the most important lower bound techniques for quantum query complexity.  In fact, a strengthening 
of the adversary bound, known as the general 
adversary bound \cite{HoyerLeeSpalek07negativeadv}, has recently been shown to characterize 
quantum query complexity, up to 
constant factors \cite{Reichardt10advtight, LMRS11}.

What we actually use in the paper, is the dual of the general adversary bound.
It provides upper bounds on the quantum query complexity, i.e., quantum query algorithms.
Due to the same results, it also is tight. 
Despite this tight equivalence, the actual applications of this upper bound (in the form of span programs) have been limited, mostly, to formulae evaluation \cite{formulae}, and, recently, linear algebra problems~\cite{spanRank}. In~\cite{spanCert}, it was used to give a variant of an optimal algorithm for the element distinctness problem, and an algorithm for the triangle problem having better complexity than the one known before. In this paper, we provide yet another application.

The (dual of the) general adversary bound is defined as follows.
\begin{definition}
Let $f\colon  [m]^n \rightarrow \{0,1\}$ be a function.  
\begin{align}
\label{eqn:adv_primal}
\begin{aligned}
\ADV^\pm(f)=&\underset{\substack{k \in \N \\ u_{x,j} \in \C^k}}{\mathrm{minimize}}\ 
\max_x \sum_{j \in [n]} \norm{u_{x,j}}^2\\
& \mathrm{subject\  to} \sum_{\substack{j \\ x_j \ne y_j}} 
\braket{u_{x,j}}{u_{y,j}} =1 \text{ whenever } f(x) \ne f(y).
\end{aligned}
\end{align}
\end{definition}

For our application, it will be more convenient to use a different formulation of the objective 
value.
\begin{clm}
\begin{align}
\label{eqn:mean_equiv}
\begin{aligned}
\ADV^\pm(f)=&\underset{\substack{k \in \N \\ u_{x,j} \in \C^k}}{\mathrm{minimize}}\ 
\sqrt{\s{\max_{x\in f^{-1}(1)}\sum_{j \in [n]} \norm{u_{x,j}}^2} \s{\max_{y\in f^{-1}(0)}\sum_{j \in [n]} \norm{u_{y,j}}^2}}.\\
& \mathrm{subject\  to} \sum_{\substack{j \\ x_j \ne y_j}} 
\braket{u_{x,j}}{u_{y,j}} =1 \text{ whenever } f(x) \ne f(y).
\end{aligned}
\end{align}
\end{clm}

\begin{proof}
The objective value in Eq.~\refeqn{mean_equiv} is less than that of Eq.~\refeqn{adv_primal} by 
the inequality of arithmetic and geometric means.  For the other direction, note that the constraint is
invariant under multiplying all vectors $u_{x,j}$ where $f(x)=1$ by $c$ and all vectors 
$u_{y,j}$ where $f(y)=0$ by $c^{-1}$.  In this way we can ensure that the maximum in 
Eq.~\refeqn{adv_primal} is the same over $f^{-1}(0)$ and $f^{-1}(1)$ and so equal to the 
geometric mean.
\end{proof}

The general adversary bound characterizes bounded-error quantum 
query complexity.
\begin{thm}[\cite{Reichardt10advtight, LMRS11}]
\label{thm:adv_characterize}
Let $f$ be as above.  Then $Q_{1/4}(f) = \Theta(\ADV^\pm(f))$.
\end{thm}

\subsection{Martingales and Azuma's Inequality}
We assume the reader is familiar with basic notions of probability theory. In this section, we state some concentration results we will need in the proof of \refthm{main}. The results are rather standard, can be found, e.g., in~\cite{probabilistic}.

A {\em martingale} is a sequence $X_0,\dots,X_m$ of random variables such that $\E[X_{i+1}\mid X_0,\dots,X_i]=X_i$, for all $i$'s.

\begin{thm}[Azuma's Inequality]
\label{thm:azuma}
Let $0=X_0,\dots,X_m$ be a martingale such that $|X_{i+1}-X_i|\le 1$ for all $i$'s. Then
\[
\pr[X_m>\lambda\sqrt{m}]<e^{-\lambda^2/2}.
\]
for all $\lambda>0$.
\end{thm}

A standard way of defining martingales, known as {\em Doob martingale process}, is as follows. Assume $f(y_1,\dots,y_m)$ if a real-valued function, and there is a probability distribution $Y$ on the input sequences. The Doob martingale $D_0,\dots,D_m$ is defined as
\[
D_i = \E_{y'\in Y}[f(y')\mid \forall j\le i: y'_j=y_j]
\]
that is a random variable dependent on $y\in Y$. In particular, $D_0=\E[f]$ and $D_m=f(y)$. This is a martingale, and Azuma's inequality states $f(y)$ isn't far away from its expectation with high probability, if revealing one input variable has little effect on the expectation of the random variable.

\section{Learning graphs}
\label{sec:learning}
\subsection{Definitions}
\label{sec:def}
By \refthm{adv_characterize}, to upper bound the quantum query complexity of a function, 
it suffices to construct a feasible solution to Eq.~\refeqn{adv_primal}.  Trying to come up 
with vectors which satisfy all pairwise equality constraints, however, can be quite challenging 
even for simple functions.  

A learning graph, introduced in~\cite{spanCert}, is a computational model that aids in the
 construction of such vectors for a function 
$f\colon[m]^n\supseteq \cD\to \{0,1\}$ with boolean output.  By design, a learning graph 
ensures that the constraint \refeqn{mean_equiv} is satisfied, allowing one to focus 
on minimizing the objective value.  


\begin{definition}
A {\em learning graph} $\cG$ is a directed acyclic connected graph with vertices labeled by subsets of 
$[n]$, the input indices. It has arcs connecting vertices $S$ and $S\cup\{j\}$ only, where 
$S\subseteq[n]$ and $j\in[n]\setminus S$. Each arc $e$ is assigned a {\em weight function} 
$w_e\colon [m]^S\to \R^+$, where $S$ is the origin of $e$.
\end{definition}

A learning graph can be thought of as modeling the development of one's knowledge about the 
input during a query algorithm. Initially, nothing is known, and this is represented by the root 
labeled by $\emptyset$.  When at a vertex labeled by $S\subseteq [n]$, the values of the variables in $S$ have been learned. Following an arc $e$ connecting $S$ to $S\cup\{j\}$ can be interpreted as querying the value of variable $x_j$. We say the arc {\em loads} element $j$. When talking about vertex labeled by $S$, we call $S$ the set of {\em loaded elements}.

In order for a learning graph to compute function $f$ correctly, for any $x\in f^{-1}(1)$, there should exist a vertex of the learning graph containing a 1-certificate for $x$. We call vertices containing 
a 1-certificate {\em accepting}.

Let $e$ be a weighted arc of the learning graph from $S$ to $S\cup\{j\}$.  In the examples 
of learning graphs given in~\cite{spanCert}, it sufficed to assign $e$ a weight $w_e$ that 
depended only on the set $S$ and element $j$, but not the values learned.  Here, we follow Remark~4 of 
\cite{spanCert} and use a more general model where $w_e$ can depend both on $S$ and $j$, as well as on 
the values of the variables in $S$. We denote $w_e(x)=w_e(x_S)$. 
Although, this notation is convenient, it is important to keep in mind that values of the variables outside $S$ {\em do not} affect the value of $w_e$. The weight 0 of an arc should be thought of as the arc is missing for this particular input. 

By $\cG(x)$, we denote the instance of $\cG$ for input $x\in\cD$, i.e., $\cG(x)$ has the same vertices and arcs as $\cG$ does, only the weight of arc $e$ is a real number $w_e=w_e(x)$. Another way to think of a leaning graph, is like a collection of graphs $\cG(x)$ such that arcs from $S$ to $S\cup\{j\}$ in $\cG(x^{(1)})$ and $\cG(x^{(2)})$ have equal weight if $x^{(1)}_S = x^{(2)}_S$. The arcs with the latter property are called {\em identical}.

Arcs is the main constituent of the learning graph, and we use notation $e\in\cG$ to denote that $e$ is an arc of $\cG$. Similarly, we write $e\in\cG(x)$.

The {\em complexity} of a learning graph computing $f$ is defined as the geometrical mean of its 
{\em positive} and {\em negative complexities}. The negative complexity $\cN(\cG(y))$ for $y\in f^{-1}(0)$ is defined as $\sum_{e\in\cG(y)} w_e$. The negative complexity of the learning graph 
$\cN(\cG)$ is defined as $\max_{y\in f^{-1}(0)} \cN(\cG(y))$. In order to define positive complexity, we need one additional notion.

\begin{definition}
\label{def:flow}
The {\em flow} on $\cG(x)$ for $x\in f^{-1}(1)$ is a real-valued function $p_e$ where $e\in\cG(x)$. It has to satisfy the following properties:
\begin{itemize}
\item vertex $\emptyset$ is the only source of the flow, and it has intensity 1. In other words, the sum of $p_e$ over all $e$'s leaving $\emptyset$ is 1;
\item vertex $S$ is a sink iff it is accepting. That is, if $S\ne\emptyset$ and $S$ does not contain a 1-certificate of $x$ for $f$ then, for vertex $S$, the sum of $p_e$ over all in-coming arcs equals the sum of $p_e$ over all out-going arcs.
\end{itemize}
\end{definition}

The complexity of the flow is defined as $\sum_{e\in \cG(x)} p_e^2/w_e$, with convention $0/0=0$. The positive complexity $\cP(\cG(x))$ is defined as the smallest complexity of a flow on $\cG(x)$. The positive complexity of the learning graph 
$\cP(\cG)$ is defined as $\max_{x\in f^{-1}(1)} \cP(\cG(x))$.

We often consider a collection of flows $p$ for all $x\in f^{-1}(1)$. In this case, $p_e(x)$ denotes the flow $p_e$ in $\cG(x)$.

Let us briefly introduce some additional concepts connected with learning graphs. The $i$-th {\em step} of a learning graph is the set of all arcs ending in a vertex of cardinality $i$. If $E\subseteq \cG$ is a set of arcs, we use notation $p_{E} = \sum_{e\in E} p_e$. Usually, $E$ is a subset of a step. A special case is $p_S$ with $S$ being a vertex; it is used to denote the flow through vertex $S$, i.e., the sum of $p_e$ over all arcs ending at $S$. 

The following technical result is extracted from~\cite{spanCert}
\begin{lem}[Conditioning]
\label{lem:cond}
Suppose $V$ is a subset of vertices such that no vertex is a subset of another. Let $p_e$ be a flow from $\emptyset$ and ending at $V$ of intensity 1, $W$ be a subset of $V$ and $t = \sum_{S\in W} p_S$. Then there exists a flow $p'$ with the same properties, such that $p'_S = p_S/t$ for $S\in W$ and $p'_S=0$, otherwise. Moreover, the complexity of $p'$ is at most $1/t^2$ times the complexity of $p$.
\end{lem}

This lemma is applied as follows. One uses some construction to get a flow that ends at $V$. After that, another construction is applied to obtain a flow that starts at $W$ and ends at the proper sinks, i.e., accepting vertices. In the second flow, vertices in $V\setminus W$ are dead-ends, i.e., no flow should leave them. Then it is possible to apply \reflem{cond} to glue both parts of the flow together and get a valid flow.

\subsection{Reduction to Quantum Query Algorithms}
In this section, we prove that if there is a learning graph for $f$ of complexity $C$ then 
$Q_{1/4}(f)=O(C)$.   In the case of $f$ with non-boolean input alphabet this solves an open 
problem from~\cite{spanCert} by removing a logarithmic factor present there.
We do this by showing how a learning graph can be used to construct a solution to 
Eq.~\refeqn{mean_equiv} of the same complexity, and then appealing to 
\refthm{adv_characterize}.
\begin{thm}
If there is a learning graph for $f\colon  [m]^n \rightarrow \{0,1\}$ with complexity $C$ then 
$\ADV^\pm(f) \le C$.
\end{thm}

\begin{proof}
Let $\cG$ be the learning graph, $w_e$ be the weight function, and $p$ be the optimal flow.

We show how to construct the vectors $u_{x,j}$ satisfying \refeqn{mean_equiv} from $\cG$. Let $E_j$ be the set of arcs 
$e_{S, S \cup \{j\}}$ between $S$ and $S \cup \{j\}$ for some $S$.  
Notice that the set of $\{E_j\}_{j \in [n]}$ partition all the arcs in the graph.  If 
$e=e_{S, S \cup \{j\}}$, let $\alpha(e) \in [m]^{S}$ be an assignment of values to the set labeling 
the origin of $e$.  

The vectors $u_{x,j}$ will live in a Hilbert space $\bigoplus_{e \in E_j,\alpha(e)} H_{e, \alpha(e)}$ 
where $\alpha(e) \in [m]^{S}$ is an assignment of values to the positions in $S$.  In our case 
each $H_{e, \alpha(e)}=\C$.  Thus we think of 
$u_{x,j}=\bigoplus_{e \in E_j, \alpha(e)} u_{x,j,e,\alpha(e)}$, and now go about designing these vectors.

First of all, if $e= e_{S, S \cup \{j\}}$ then $u_{x,j,e, \alpha(e)}=0$ if 
$x_S \ne \alpha(e)$.  Otherwise,  if $f(y)=0$ then we set 
$u_{y,j,e,\alpha(e)}= \sqrt{w_e(y)}$ and if $f(x)=1$, we set $u_{x,j,e,\alpha(e)}=p_e(x)/\sqrt{w_e(x)}$.     

Let us check the objective value.  If $f(y)=0$ then we have 
\[
\sum_j \norm{u_{y,j}}^2 =\sum_j \sum_{e \in E_j} w_e(y) = \sum_{e\in \cG} w_e(y) = \cN(\cG(y)).
\]
If $f(x)=1$ then 
\[
\sum_j \norm{u_{x,j}}^2=\sum_j \sum_{e \in E_j} \frac{p_e(x)^2}{w_e(x)} = \sum_{e\in\cG} \frac{p_e(x)^2}{w_e(x)} = \cP(\cG(x)).
\]
Thus the geometric mean of these quantities it is at most $C$.  

Let us now see that the constraint is satisfied. 
\begin{align*}
\sum_{j:  x_j \ne y_j} \braket{u_{x,j}}{u_{y,j}} &= 
\sum_{j:  x_j \ne y_j} \sum_{\substack{ e_{S, S \cup \{j\}} \\ x_S=y_S}} 
\braket{u_{x,j,e,x_S}}{u_{y,j,e,x_S}} \\ 
&= \sum_j \sum_{\substack{ e_{S, S \cup \{j\}} \\ x_S=y_S, x_j \ne y_j}} \frac{p_e(x)}{\sqrt{w_e(x)}}\sqrt{w_e(y)} = \sum_j \sum_{\substack{ e_{S, S \cup \{j\}} \\ x_S=y_S, x_j \ne y_j}} p_e(x) = 1 \enspace.
\end{align*}
The second equality from the end holds because $w_e(x) = w_e(x_S) = w_e(y_S) = w_e(y)$ due to the construction of the weight function. 
To see why the last equality holds, note that the set of arcs  
$e_{S, S \cup {j}}$ where $x_S=y_S$ and $x_j \ne y_j$ is the cut induced by the vertex sets $\{S \mid x_S=y_S\}$ and $\{S\mid x_S\ne y_S\}$. Since the source is in the first set, and all the sinks are in the second set, the value of the cut is equal to the total flow which is one. \end{proof}

\section{Getting Ready}
\label{sec:getting}
This section is devoted to the analysis of the applicability of learning graphs for the $k$-distinctness problem, without constructing the actual learning graph. In \refsec{sym}, we review the tools of~\cite{spanCert} to the case when the arcs of the learning graph depend on the values of the variables. In \refsec{loose}, we make the tools of in \refsec{sym} easier to apply. In \refsec{convention}, we describe the conventions on the input variables we assume for the rest of the paper. In \refsec{almost}, we develop an important notion of almost symmetric flow that is a generalization of symmetric flow used in~\cite{spanCert}. Finally, in \refsec{previous}, we describe a learning graph that is equivalent to the previous quantum algorithm for the $k$-distinctness problem.

\subsection{Symmetries}
\label{sec:sym}
In~\cite{spanCert}, the symmetries under consideration were those of the indices of the input variables. This was sufficient because values of the variables did not affect the learning graph. In this paper, we consider a wider group of symmetries, namely $\cS_n\times \cS_m$, where $\cS$ is the full symmetric group, that in the first multiplier permutes the indices, and the second one the values of the variables, i.e., an input $x = (x_i)_{i\in[n]}$ gets mapped by $\sigma = \si\times\so$ to $\sigma x = (\so (x_{\si i}))_{i\in [n]}$. 

Let $\Sigma\subseteq \cS_n\times \cS_m$ be the symmetry group of the problem, i.e., such that $f(\sigma x) = f(x)$ for all $x\in\cD$ and  $\sigma\in\Sigma$. For the $k$-distinctness problem, $\Sigma$ equals the whole group $\cS_n\times \cS_m$.

We extend the mapping $x\mapsto\sigma x$ to assignments, as well as vertices and arc of learning graphs in an obvious way. For example, an arc $e\in \cG(x)$ from $S$ to $S\cup\{v\}$ is mapped to the arc $\sigma e\in \cG(\sigma x)$ from $\si S$ to $\si (S\cup\{ v\})$. Actually, graph $\cG(\sigma x)$ may also not contain the latter arc. To avoid such inconvenience, we assume $\cG$ is embedded into the complete graph having all possible arcs of the form $e_{S,S\cup\{v\}}$, with the unused arcs having weights and flow equal to 0. Then, it is easy to see any $\sigma\in\Sigma$ maps a valid flow on $\cG(x)$ to a valid flow on $\cG(\sigma x)$ in the sense of \refdef{flow}. Of course, the complexity of the latter can be huge, even $+\infty$, because it may have a non-zero flow through an arc having weight 0. Consider two arcs:
\begin{equation}
\label{eqn:arcs}
\text{$e_i\in \cG(x^{(i)})$ originating in $S_i$ and loading $v_i$, for $i=1,2$.}
\end{equation}
In this section, as well as in \refsec{loose}, we are going to define various equivalence relations between them, mostly, to avoid the increase in the complexity of the flow under transformations from $\Sigma$. Note, in contrary to~\cite{spanCert}, we define equivalences between arcs, not transitions, i.e., chains of arcs.

\paragraph{Equivalency}
Arcs $e_1$ and $e_2$ are called {\em equivalent} iff there exists $\sigma\in\Sigma$ such that $\si(v_1)=v_2$ and $\sigma(x^{(1)}_{S_1})=x^{(2)}_{S_2}$. It is natural to assume equivalent arcs have equal weight. We give a formal argument in \refprp{sym_best}.

Denote by $\cE$ the set of all equivalency classes of $\cG$ under this relation. Also, we use notation $\cE_i$ for all equivalency classes of step $i$ (an equivalency class is fully contained in one step, hence, this is a valid notion). If $E\in\cE$, we use notation $E(x)$ to denote the subset of arcs of $\cG(x)$ that belongs to $E$.

For the $k$-distinctness, the equivalence is characterized by the structure of the subtuples of $S$. We capture this by the {\em specification} $\beta(S)$ of the vertex, i.e., by a list of non-negative integers $(b_1,b_2,\dots,b_{k-1})$ such that $S$ contains exactly $b_t$ $t$-subtuples. (If $S$ contains a $k$- or a larger subtuple it is an accepting vertex and no arcs are leaving it). In particular, $|S|=\sum_t tb_t$. Thus, two arcs are equivalent iff the specifications of their origins are equal.

\paragraph{Strong equivalency} The first part of this section describes the equivalence relation for arcs $e_1$ and $e_2$ with respect to their weight. We would like to get a stronger equivalence that captures the flow through an arc. This kind of equivalency has already been used in~\cite{spanCert} without explicit definition. 

Arcs $e_1$ and $e_2$ from~\refeqn{arcs} are called {\em strongly equivalent} iff there exists an element $\sigma\in\Sigma$ such that 
\begin{equation}
\label{eqn:strong}
\sigma(x^{(1)})=x^{(2)},\qquad\si(S_1)=S_2\qquad\mbox{and}\qquad \si(v_1)=v_2.
\end{equation}

Again, due to symmetry, it is natural to assume the flow through strongly equivalent arcs is equal. If, for some positive inputs $x^{(1)}$ and $x^{(2)}$, there is $\sigma\in\Sigma$ such that the first condition of~\refeqn{strong} holds, the task of finding a flow for $x^{(2)}$ is reduced to finding a flow for $x^{(1)}$, that is again a corollary of \refprp{sym_best}. See also \refprp{substitute}.

\paragraph{Formal argument} We give a proof that, without loss in complexity, we may assume weight and flow is constant on equivalent and strongly equivalent arcs, respectively.
\begin{prp}
\label{prp:sym_best}
For any learning graph $\cG$, it is possible to construct a learning graph $\cG'$ and a flow $p'$ on it with the same or smaller complexity, so that equivalent arcs have the same weight and strongly equivalent arcs have the same flow through them.
\end{prp}

\pfstart
The proof is a standard application of symmetry. Let $p$ be an optimal flow for $\cG$. We define the weights of arcs in $\cG'$ and the flow through it as follows:
\[
w_e'(\alpha) = \frac1{|\Sigma|}\sum_{\sigma\in\Sigma} w_{\sigma e}(\sigma\alpha),\qquad\text{and}\qquad p_e'(x) = \frac1{|\Sigma|}\sum_{\sigma\in\Sigma} p_{\sigma e}(\sigma x).
\]
If arcs of~\refeqn{arcs} are equivalent, there exists $\sigma'$ such that $\sigma' e_1=e_2$ and $\sigma'(x^{(1)}_{S_1})=x^{(2)}_{S_2}$. Hence,
\[
w_{e_2}'(x^{(2)}) = \frac1{|\Sigma|}\sum_{\sigma\in\Sigma} w_{\sigma e_2}(\sigma(x^{(2)}_{S_2})) = 
\frac1{|\Sigma|}\sum_{\sigma\in\Sigma} w_{\sigma\sigma' e_1}(\sigma\sigma' (x^{(1)}_{S_1})) = w_{e_1}'(x^{(1)}),
\]
since $\Sigma$ is a group. The equality of flows is proven in a same way. Let us check the complexity. For a negative input $y$, we have:
\[
\cN(\cG'(y))= \sum_{e\in\cG} \frac1{|\Sigma|}\sum_{\sigma\in\Sigma} w_{\sigma e}(\sigma y) = \frac1{|\Sigma|}\sum_{\sigma\in\Sigma}\cN(\cG( \sigma y)).
\]
Hence, for at least one $\sigma$, $\cN(\cG'(y))\le \cN(\cG(\sigma y))$. Thus, $\cN(\cG')\le \cN(\cG)$.

For the positive case, at first note that $p'$ is a valid flow, as a convex combination of valid flows. For any $x\in f^{-1}(1)$, we have
\begin{align*}
\cP(\cG'(x)) \le \sum_{e\in\cG} &\s{\frac1{|\Sigma|}\sum_{\sigma\in\Sigma} p_{\sigma e}(\sigma x)}^2 \s{\frac1{|\Sigma|}\sum_{\sigma\in\Sigma} w_{\sigma e}(\sigma x)}^{-1}\\
&\le \sum_{e\in\cG} \frac1{|\Sigma|} \sum_{\sigma\in\Sigma} \frac{p_{\sigma e}(\sigma x)^2}{w_{\sigma e}(\sigma x)}=\frac1{|\Sigma|} \sum_{\sigma\in\Sigma}\cP(\cG(\sigma x)).
\end{align*}
The second inequality follows from the Jensen's inequality for the square function
$
\s{\sum_{\sigma\in\Sigma} \gamma_\sigma z_\sigma}^2 \le \sum_{\sigma\in\Sigma} \gamma_\sigma z_\sigma^2,
$
with $\gamma_\sigma = w_{\sigma e}(\sigma x)/\s{\sum_{\sigma\in\Sigma} w_{\sigma e}(\sigma x)}$ and $z_\sigma = p_{\sigma e}(\sigma x)/\gamma_\sigma$. Due to the same argument, $\cP(\cG')\le \cP(\cG)$.
\pfend

\subsection{Loosening equivalencies}
\label{sec:loose}
Although equivalencies defined in the previous section are optimal, they are not always convenient to work with. They turn out to be too strong, that results in a vast number of equivalency classes that should be treated separately. In this section, we describe a number of ways to loosen these equivalences, thus reducing the number of classes and making them easier to work with. 

\paragraph{Equivalency} Assume the weight function is decomposed as $w_e(\alpha) = w_e(\theta(\alpha))$, where $\theta$ is some ``filter'' that captures the properties of $\alpha$ we are interested in. It is good to assume symmetry preserves $\theta$, i.e., $\theta(\alpha_1)=\theta(\alpha_2)$ implies $\theta(\sigma\alpha_1)=\theta(\sigma\alpha_2)$ for any $\sigma\in\Sigma$. Transitions $e_1$ and $e_2$ are called {\em $\theta$-equivalent} iff there exists $\sigma\in\Sigma$ such that $\si(v_1)=v_2$ and $\sigma(\theta(x^{(1)}_{S_1}))=\theta(x^{(2)}_{S_2})$. It is again natural to assume $\theta$-equivalent arcs have the same weight.
Two main examples are:
\begin{itemize}
\item[$\theta_1$,] the identity. This results in the relation from the previous section. This is the main equivalency used in \refsec{main};
\item [$\theta_2$,] mapping $\alpha$ to its domain $\cD(\alpha)$. For the $k$-distinctness, arcs $e_1$ and $e_2$ from~\refeqn{arcs} are $\theta_2$-equivalent iff $|S_1|=|S_2|$. This is the equivalency used in \cite{spanCert}.
\end{itemize}  

\paragraph{Strong equivalency} Unlike equivalency, strong equivalency turns out to be too strong for all our applications. We can weaken it by considering $\cG$ as a learning graph for function $\tilde f$ that gets as input $\tilde x$, the original input $x$ with some information removed.

More precisely, extend the output alphabet $[m]$ with a set of {\em special characters} $Q$. Let $\vartheta\colon f^{-1}(1)\to ([m]\cup Q)^n$ be the function that maps $x$ to $\tilde x$. We extend elements of $\Sigma$ to $([m]\cup Q)^n$ by assuming $\so(c)=c$, if $\si\times\so \in \Sigma$ and $c\in Q$. 

The function $\tilde f\colon ([m]\cup Q)^n\to\{0,1\}$ is defined as $\tilde f(\vartheta(x))=1$ for all $x\in f^{-1}(1)$. Let $\tilde\cG$ be the same learning graph as $\cG$ but calculating $\tilde f$. Arcs $e_1$ and $e_2$ from $\refeqn{arcs}$ are called {\em $\vartheta$-strongly equivalent} iff the corresponding arcs $\tilde e_1\in\tilde\cG(\vartheta(x^{(1)}))$ and $\tilde e_2\in\tilde\cG(\vartheta(x^{(2)}))$ are strongly equivalent. 

For this construction to work, we require a stronger definition of a 1-certificate for $\tilde f$. We say an assignment $\alpha\colon [n]\supseteq M\to [m]\cup Q$ is a {\em 1-certificate} if, for all $x\in f^{-1}(1)$ such that $\vartheta(x)_M=\alpha$, $x_M$ is a 1-certificate of $x$ in $f$. With this definition, any valid flow in $\tilde\cG(\vartheta(x))$ is simultaneously a valid flow in $\cG(x)$.

We give three examples of $\vartheta$'s.
\begin{itemize}
\item[$\vartheta_1$,] the identity. This results in the relation from the previous section. We have no example of using this equivalency;
\item [$\vartheta_2$,] the equivalency used in \cite{spanCert}. Let $Q=\{\cdot, \star\}$. For a positive input $x$, fix some $1$-certificate $\alpha$. Let $M$ be the domain of $\alpha$. The elements of $M$ are called {\em marked}. Define $\tilde x=\vartheta_2(x)$ as
\[
\tilde x_i = \begin{cases}
\star,&i\in M;\\
\cdot,&\text{otherwise.}
\end{cases}
\]
Clearly, $\vartheta_2(x)_M$ is a 1-certificate. Refer to \refsec{previous} for an example of usage of this equivalency.
\item [$\vartheta_3$,] defined in \refsec{convention}. The main equivalency used in \refsec{main}.
\end{itemize}

Again, we assume the flow through $\vartheta$-strongly equivalent arcs is equal. The equivalencies we use in the paper possess two additional symmetric properties. Firstly, if $x^{(1)},x^{(2)}\in f^{-1}(1)$ and $\sigma\in\Sigma$ are such that $\sigma(\vartheta(x^{(1)}))=\vartheta(x^{(2)})$ then, for any 1-certificate $\alpha$ of $\vartheta(x^{(1)})$, $\sigma\alpha$ is a 1-certificate of $\vartheta(x^{(2)})$. Secondly, any $\vartheta$-strong equivalency class, having non-zero flow through it, is completely contained in some $\theta$-equivalency class. 
\begin{prp}
\label{prp:substitute}
If $\theta$, $\vartheta, x^{(1)},x^{(2)}$ and a flow $p_e$ on $\tilde\cG(\vartheta(x^{(1)}))$ satisfy the above conditions and $\theta$-equivalent arcs in $\cG$ have the same weight, then $\sigma p_e$ is a valid flow on $\tilde\cG(\vartheta(x^{(2)}))$ with the same complexity (that, also, is a valid flow on $\cG(x^{(2)})$).
\end{prp}

\subsection{Conventions for $k$-distinctness}
\label{sec:convention}
Strong equivalency, as defined in \refsec{sym}, turns out to be too strong for the $k$-distinctness problem, because for most of the pairs $x^{(1)},x^{(2)}\in f^{-1}(1)$ there is no $\sigma$ such that $\sigma(x^{(1)})=x^{(2)}$, and, hence, no arcs from $\cG(x^{(1)})$ and $\cG(x^{(2)})$ can be strongly equivalent, whatever $\cG$ is. We use the loosening tool of \refsec{loose} to define $\vartheta_3$ so that there always exists $\sigma$ that maps $\vartheta_3(x^{(1)})$ to $\vartheta_3(x^{(2)})$. Then, by \refprp{substitute}, defining a flow for any positive input $x$ is enough to get a flow for all positive inputs.

Let the set of special characters be $Q=\{\cdot\}$. Fix an arbitrary positive input $x$. First of all, we identify a subset $M$ of $k$ equal elements (the marked elements in terminology of $\vartheta_2$). Next, due to the condition in \refthm{main}, we may assume there are non-negative integers $\ell_1,\dots,\ell_{k-1}$ such that in any valid input (either positive, or negative) there are at least $\ell_t$ $t$-tuples and 
\[ n-\sum_{t=1}^{k-1}t\ell_t = O(\sqrt[4]{n}). \]
We arbitrary select $\ell_t$ $t$-tuples. Denote by $A_t$ the union of the selected $t$-tuples. We also use notation $A_{\ge t}$ to denote $\bigcup_{j=t}^{k-1} A_j$. We define $\tilde x = \vartheta_3(x)$ as
\[
\tilde x_i = \begin{cases}
x_i,&\text{$i\in A_{\ge1}\cup M$;}\\
\cdot,&\text{otherwise.}
\end{cases}
\]
Clearly, assignment $\tilde x_M$ is a 1-certificate. In the learning graph for $\tilde f$, defined using $\vartheta_3$, we will have $p_e(x)=0$ if the origin of $e$ has at least one $\cdot$ in $\vartheta_3(x)$. This convention assures that $\theta_1$ and $\vartheta_3$ satisfy the conditions of \refprp{substitute}. Further, we are going to ignore vertices having $\cdot$'s in them. \reffig{reduction} describes which of the defined equivalence relations imply which.

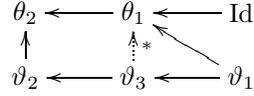
\begin{figure}[tb]
$$\xymatrix@R=10pt{
\theta_2  &\theta_1\ar[l] & \text{Id}\ar[l] \\
\vartheta_2\ar[u] & \vartheta_3\ar[l]\ar@{.>}[u]_{*} & \vartheta_1\ar[l]\ar[lu]
}$$
\caption{Relations between equivalencies. Arrows are from a more strong relation to a weaker one. Id is the identity relation from \refsec{def}, $\theta$'s are equivalencies and $\vartheta$'s are strong equivalencies. Implication from $\vartheta_3$ to $\theta_1$ only holds for arcs with non-zero flow.} \label{fig:reduction}
\end{figure}

Let us use this spot to mention one more convention on the input. Namely,
\begin{equation}
\label{eqn:lk-1}
\forall t\le k-1: \ell_t = \Omega(n).
\end{equation}
Any other case can be reduced to this one by extending the input by $n$ $t$-tuples with elements outside the range of the original problem, for each $t\le k-1$. 

The strong equivalency class (with respect to $\vartheta_3$) of an arc depends solely on the {\em types} of its initial and target vertices. The type $\tilde\beta(S)$ of vertex $S$ is an $(k-1)\times k$-matrix $(b_{t,s})$, where $b_{t,s}$ is the number of $t$-subtuples of $S$ contained in $A_s$ (or $M$, if $s=k$). Most of the time, we will implicitly assume $b_{t,k}=0$ for all $t$'s, hence, describe the type of a vertex by an $(k-1)\times(k-1)$-matrix, assuming the removed row contains only zeroes. Note also that the specification $(b_t)$ can be expressed using the type: $b_t = \sum_{s=1}^{k} b_{t,s}$. 

We will have to measure distance between types. When doing so, we treat them as vectors. I.e., the distance between types $\tilde\beta(S)=(b_{s,t})$ and $\tilde\beta(S') = (b'_{s,t})$ is defined as 
\[
\|\tilde\beta(S)-\tilde\beta(S')\|_\infty = \max_{s,t} |b_{s,t}-b'_{s,t}|.
\]

\paragraph{Negative complexity} Here we estimate how the restriction from the actual number of $t$-tuples in the input to $\ell_t$ ones in $A_{\ge1}$ affects the negative complexity. 
\begin{lem}
\label{lem:Aprim}
Consider a set $A'_{\ge 1}$ that is defined similarly to $A_{\ge 1}$, only it has $\ell_t'$ $t$-tuples. Assume $|\ell_t-\ell'_t|\le d=o(n)$ for all $t$'s. Let $(b_t)$ be any specification such that $b=\sum_t b_t=o(n)$. Then the ratio of the number of subsets satisfying $(b_t)$ in $A_{\ge 1}$ and $A'_{\ge 1}$ is at most $e^{O(db/n)}$.
\end{lem}

\pfstart It is straight-forward to calculate the number of subsets of $A_{\ge1}$ satisfying specification $(b_t)$. Indeed, it equals
\begin{equation}
\label{eqn:negative1}
\sum_{(b_{t,s})} \prod_{s=1}^{k-1} \left[{\ell_s \choose b_{1,s}}{s\choose 1}^{b_{1,s}}{\ell_s-b_{1,s}\choose b_{2,s}}{s\choose 2}^{b_{2,s}}\cdots{\ell_s - b_{1,s} - \cdots - b_{s-1,s}\choose b_{s,s}}{s\choose s}^{b_{s,s}} \right]
\end{equation}
where the summation is over all types $(b_{t,s})$ that agree with specification $(b_t)$. Eq.~\refeqn{negative1}, with $\ell_t$ replaced by $\ell_t'$, gives the corresponding number of subsets in $A'_{\ge 1}$. It is enough to show that each multiplier featuring $\ell_s$ in \refeqn{negative1} changes by at most a factor of $e^{O(db/n)}$. But we have:
\[{\ell'_s - b_{1,s} - \cdots - b_{t-1,s}\choose b_{t,s}}/{\ell_s - b_{1,s} - \cdots - b_{t-1,s}\choose b_{t,s}}
=\s{1+O\s{\frac dn}}^{O(b)} = e^{O(db/n)},\]
where we used that $\ell_s = \Theta(n)$, because of~\refeqn{lk-1}.
\pfend

Since the complexity mentioned in \refthm{main} is $o(n^{3/4})$, it is natural to assume no vertex of the learning graph has more elements. It's actually the case, as described in \refsec{main}. The precision $O(\sqrt[4]{n})$ in the formulation of \refthm{main} has been chosen so that restriction of the flow to $A_{\ge1}\cup M$ does not hurt the negative complexity, as it can be seen from the next

\begin{cor}
\label{cor:negative}
Fix any possible negative input $y$, and any valid specification $(b_t)$ with all entries $o(n^{3/4})$. Then, the number of subsets of $[n]$ satisfying $(b_t)$ is bounded by a constant times the number of such subsets included in $A_{\ge 1}$.
\end{cor}

Because of this, we may act as if the set of input variables is $A_{\ge1}\cup M$, not $[n]$.

\subsection{Almost symmetric flows}
\label{sec:almost}
Assume the following scenario. We have chosen which equivalency classes will be present in the learning graph. Also, for each positive input, we have constructed a flow. The task is to weight the arcs of the learning graph to minimize its complexity. In this section, we define a way of performing this task, if the flow satisfies some requirements.

For the $k$-distinctness problem, all arcs leaving a vertex are equivalent, hence, to specify which equivalency classes are present, it is enough to define which vertices have arcs leaving them. For each step, we define a set of {\em valid specifications}. If a vertex before the step satisfies one of them, we draw all possible arcs out of it. Otherwise, we declare it a {\em dead-end} and draw no arcs out of it. 

The flow is called {\em symmetric} in~\cite{spanCert} if, for each equivalency class, the flow through an arc of it is either 0, or $p$, where $p$ does not depend on the input, but may depend on the equivalency class; also it is required that the number of arcs having flow $p$ does not depend on the input as well. This notion was sufficient for the applications in that paper, because $\vartheta_2$-strong equivalence was used, and that is easy to handle. In this paper, we use $\vartheta_3$-strong equivalence, and it is not enough with symmetric flows. Thus, we have to generalize this notion.

\begin{definition}
The flow is called {\em almost symmetric} if, for each equivalency class $E$, there exist constants $\pi(E)$ and $\tau(E)$ such that, for each positive input $x$, there exists a subset $G(E, x)\subseteq E(x)$ such that
\begin{equation}
\label{eqn:almost}
\begin{split}
 &\tau(E)|G(E, x)| = \Theta\s{\max_{y\in f^{-1}(0)} |E(y)|},\qquad \sum_{e\in G(E, x)} p_e(x)^2 = \Omega \s{\sum\nolimits_{e\in E(x)} p_e(x)^2} \\& \qquad\qquad\mbox{and}\quad \forall e\in G(E, x): p_e(x) = \Theta(\pi(E)).
\end{split}
\end{equation}
\end{definition}
The elements inside $G(E, x)$ are called {\em typical} arcs. Number $\tau(E)$ is called the {\em speciality} of the equivalency class (as well, as of any arc in the class). We also define the {\em typical flow} through $E$ as $\mu(E) = \pi(E)\max_{x\in f^{-1}(1)} |G(E,x)|$. It is straight-forward to check that
\begin{equation}
\label{eqn:typical_flow}
\forall x\in f^{-1}(1): \mu(E) = O(p_E(x)).
\end{equation}

\begin{thm}
\label{thm:symmetric}
If the flow is almost symmetric, the learning graph can be weighted so that its complexity becomes $O\s{\sum_{E\in\cE} \mu(E)\sqrt{\tau(E)}}$.
\end{thm}

\pfstart
For each arc $e$ in an equivalency class $E$, we assign weight $w_e = \pi(E)/\sqrt{\tau(E)}$. Let us calculate the complexity. For each $y\in f^{-1}(0)$, we have the following negative complexity
\[ \sum_{E\in \cE} w_E |E(y)| = \sum_{E\in \cE} \frac{\pi(E)}{\sqrt{\tau(E)}} |E(y)| = O\s{ \sum_{E\in \cE} \pi(E)\sqrt{\tau(E)}\max_{x\in f^{-1}(1)} |G(E,x)| }. \]
For a positive input $x\in f^{-1}(1)$, we have
\[ \sum_{E\in \cE} \frac1{w_E} \sum_{e\in E(x)} p_e(x)^2 = O\s{\sum_{E\in \cE} \frac{\sqrt{\tau(E)}}{\pi(E)} |G(E,x)|\pi(E)^2} = O\s{\sum_{E\in\cE} \mu(E)\sqrt{\tau(E)}}.\]
By combining both estimates, we get the statement of the theorem.
\pfend

For each step $i$, define $T_i = \max_{E\in\cE_i} \tau(E)$. Then \refthm{symmetric} together with~\refeqn{typical_flow} and the observation that the total flow through all arcs on any step is at most 1, implies the following
\begin{cor}
\label{cor:symmetric}
If the flow is almost symmetric, the learning graph can be weighted so that its complexity becomes $O\s{\sum_i \sqrt{T_i}}$ where the sum is over all steps.
\end{cor}

\subsection{Previous Algorithm for $k$-distinctness}
\label{sec:previous}
As an example of application of \refcor{symmetric}, we briefly describe a variant of a learning graph for the $k$-distinctness problem. It is a direct analog of an algorithm from \cite{distinct} using learning graphs and a straightforward generalization of the learning graph for element distinctness from \cite{spanCert}.

To define equivalencies between arcs, we use $\theta_2$ and $\vartheta_2$ from \refsec{loose}. The learning graph consists of loading $r+k$ elements without any restrictions (as imposed by $\theta_2$), where $r$ is some parameter to be specified later. We refer to the first $r$ steps as to the {\em first stage}, and to the last $k$ steps as to the {\em second stage}.

Clearly, all arcs of the same step are equivalent. Consider strong equivalency. Let $x$ be a positive input and let $M$ be a subset of $k$ equal elements in it. We use $M$ as the set of marked elements to define $\vartheta_2$. Then, the strong equivalence class of an arc is determined by the number of elements in its origin, the number of marked elements among them, and whether the element being loaded is marked.

The flow is organized as follows. On the first stage, only arcs without marked elements are used. On the second stage, only arcs loading marked elements are used. Thus, on each step only one strong equivalency class is used, and the flow among all arcs in it is equal. 

It is easy to check this is a valid flow for $k$-distinctness and it is symmetric. Let us calculate the specialities. The first $r$ steps have speciality $O(1)$. The speciality of the $i$-th step of the second stage is $O(n^i/r^{i-1})$. This is because the fraction of $(r+i-1)$-subsets of $[n]$ containing $i-1$ marked elements is $\Theta(r^{i-1}/n^{i-1})$; and $k-i+1$ arc only, out of $\Theta(n)$ originating in such vertex, is used by the flow. Hence, by \refcor{symmetric}, the complexity of the learning graph is $O\s{r+\sqrt{n^k/r^{k-1}}}$ that is optimized when $r=n^{k/(k+1)}$ and the optimal value is $O\s{n^{k/(k+1)}}$.

\section{Algorithm for $k$-distinctness}
\label{sec:main}
The purpose of this section is to prove \refthm{main}. In \refsec{intuition}, we give some intuition behind the learning graph. In \refsec{description}, we describe the learning graph, or, more precisely, define valid specifications for each step, as described in \refsec{convention}. In \refsec{flow}, we define the flow, and give preliminary estimates of the complexity. Finally, in \refsec{analysis}, we prove the flow defined in \refsec{flow} is almost symmetric and prove the estimates therein are correct.

\subsection{Intuition behind the algorithm}
\label{sec:intuition}
There is another way to analyze the complexity of the learning graph in \refsec{previous}.

\begin{lem}
\label{lem:subtuples}
Assume convention~\refeqn{lk-1} on the input. The expected number of $t$-subtuples in an $r$-subset of $A_{\ge1}$, chosen uniformly at random, is $\Theta(r^t/n^{t-1})$.
\end{lem}

\pfstart
Let $S$ be the random subset. Denote $n'=|A_{\ge1}|$. The probability a fixed subset of $t$ equal elements from $A_s$ forms a $t$-subtuple in $S$ is ${n'-s\choose r-t}/{n'\choose r}=\Theta(r^t/n^t)$. The number of such subsets is $\sum_s\ell_s{s\choose t} = \Theta(n)$. Hence, by linearity of expectation, the expected number is $\Theta(r^t/n^{t-1})$.
\pfend

Consider the following informal argument. Let $M$ be the set of marked elements as in \refsec{previous}. Before the last step, the flow only goes through vertices $S$ having $|S\cap M|=k-1$. Fix a vertex $S$ and let $M'=M\cap S$. One may say, $M'$ as a $(k-1)$-subtuple, is hidden among other $(k-1)$-subtuples of $S$. The expected number of such is $\Theta(r^{k-1}/n^{k-2})$, total number of $(k-1)$-tuples is $\Theta(n)$, hence, the fraction of the vertices used by the flow on this step is $\Theta(r^{k-1}/n^{k-1})$. Thus, the speciality of the arc loading the missing marked element is $\Theta(n^{k}/r^{k-1})$ that equals the estimate in \refsec{previous}.

As such, this is just a more difficult and less strict analysis of the learning graph. But one can see that the speciality of the last steps depends on the number of $t$-subtuples in the vertices. We cannot get a large quantity of them by loading elements blindly without restrictions, but it is quite possible, we can deliberately enrich vertices of the learning graph in large subtuples by gradually filtering out vertices containing a small number of them.

\subsection{Description of the Learning graph}
\label{sec:description}
We would like to apply \refcor{symmetric}, hence, it is enough to give valid specifications for each step. We do this using a pseudo-code notation in \refalg{kdist}. 

\begin{algorithm}
\caption{Learning graph for the $k$-distinctness problem}
\label{alg:kdist}
\begin{algorithmic}[1]
\For {$j\gets 1$ \textbf{to} $r_1$} \label{line:1start}
	\State \label{line:1} Load an element
\EndFor \label{line:1end}
\State Declare as dead-ends vertices having more than $c_t r_1^t/n^{t-1}$ $t$-subtuples for any $t=2,\dots,k-1$ \label{line:deadends}
\For {$i\gets 2$ \textbf{to} $k-1$} \label{line:2start}
	\For {$j\gets 1$ \textbf{to} $r_i$}
		\For {$l\gets 1$ \textbf{to} $i$}
			\State  Load an element of level $l$ \label{line:2}
		\EndFor
	\EndFor
\EndFor \label{line:2end}
\State Load an element \label{line:last} \Comment{The last element is no subject to any constraints} 
\end{algorithmic}
\end{algorithm}

Here $r_1,\dots, r_{k-1}$ are some parameters with $r_{i+1}=o(r_i)$, $r_1 = o(n)$ and $r_{k-1} = \omega(1)$ to be specified later. Also, it will be convenient to denote $r_0 = n$. The commands of the algorithm define the specifications as follows. The loop in lines~\ref{line:1start}---\ref{line:1end} says there is no constraint on the first $r_1$ steps. \refline{deadends} introduces the {\em original specifications}. Here, $c_t>0$ are some constants we specify later. 

The loop in Lines~\ref{line:2start}---\ref{line:2end} describes how the specifications change with each step. Assume a step, described on \refline{2}, loads an element of {\em level $l$}. Then, a valid specification $(b_t)$ before the step is transformed into a valid specification $(b'_t)$ after the step as follows
\[ b'_t = \begin{cases} b_t+1,& t=l;\\ b_t-1,& t=l-1; \\ b_t,&\mbox{otherwise.} \end{cases} \]
In other words, if there is an arc between vertices of specifications $(b_t)$ and $(b'_t)$ and it load $v$ then there exists an $(l-1)$-subtuple $Q$ of $S$ such that $Q\cup\{v\}$ is an $l$-subtuple of $S\cup\{v\}$. In fact, only such arcs will be used by the flow, as it is described in more detail in \refsec{flow}.

Hence, for each specification in Lines~\ref{line:2start}---\ref{line:last}, it is possible to trace it back to its original specification. For example, if $(b_t)$ is a specification of the vertex after step in \refline{2} with the values of the loop counters $i,j$ and $l$, the original specification is given by $(\tilde b_t)-(\delta^{l}_t)$, where
\[\tilde b_t = \begin{cases} b_t-r_t,& 2\le t<i;\\ b_t-j+1,& t=i; \\ b_t,&\mbox{otherwise;} \end{cases}
\qquad\mbox{and}\qquad \delta^{l}_t=\begin{cases}1,&t=l;\\0,&\mbox{otherwise.}\end{cases}
 \]
Moreover, the use of the arcs in the flow, as described in the previous paragraph, implies the flow through all vertices having some fixed original specification is the same for all steps.

Finally, the step on \refline{last} loads the last element, and there is no need for the dead-end conditions, because after the last step all vertices have no arcs leaving them.

\paragraph{Naming convention} 
We use the following convention to name the steps of the learning graph. The step on \refline{1} is referred as the $j$-th step of the first stage. The step on \refline{2} is referred using triple $(i,j,l)$, except for the case $i=k-1$ and $j=r_{k-1}$. The latter together with the step on line \refline{last} is referred as the steps $1,2,\dots,k$ of the {\em last stage}. The steps of the form $(i,\cdot,\cdot)$ are called the {\em $i$-th stage}. Altogether, all steps of the form $(\cdot,\cdot,\cdot)$ are called the {\em preparatory phase}.

\subsection{Flow}
\label{sec:flow}
We two possible ways to define a flow. The first one is to set the flow through the arcs on each step so that the flow through all vertices on each step is the same. We believe this can be done, but we lack techniques to deal with this kind of arguments.

Instead of that, we select the second way. For each vertex, we divide the flow evenly among all possible arcs. Because of this, the ratio of the maximal and the minimal flow accumulates with each step, and at the end it is quite large. We avoid this complication by applying the concentration results stating that for large $n$'s almost all flow will be concentrated on some typical subset of arcs and will be distributed almost evenly on it.

\paragraph{First Stage} For the first stage, we use $\theta_2$- and $\vartheta_2$-based equivalencies, akin to the first stage of the flow in \refsec{previous}. Consider the uniform flow, i.e., such that distributes all the in-coming flow among all out-going arcs, leading to an element of $A_{\ge1}$, equally. Clearly, it is symmetric, and the flow through any vertex $S\subseteq A_{\ge1}$ after the first stage is ${|A_{\ge1}|\choose r_1}^{-1}$. The speciality of each step in this flow is $O(1)$ because of \refcor{negative}.

But this flow is non-zero for vertices declared as dead-ends in \refline{deadends}. We fix this by applying \reflem{cond}. We have to choose $c_t>0$ so that, with probability, say, 1/2, an uniformly picked subset of size $r_1$ satisfies a valid specification. And it is possible to do so due to \reflem{subtuples} and Markov's inequality.

After performing the conditioning, the complexity of the flow in the first stage increases by at most a constant factor (that can be ignored), and all non-dead-end vertices have the same flow through them, we denote $p_o$.

\paragraph{Preliminary Estimates} For the remaining stages, we use $\theta_1$ and $\vartheta_3$ to define equivalences between arcs.
Here we informally analyze the flow for Lines~\ref{line:2start}---\ref{line:last} of \refalg{kdist}, assuming there is flow $p_o$ through all non-dead-end vertices after \refline{deadends}. The formal analysis is done in \refsec{analysis}.

Roughly speaking, the flow is organized as follows. On step $(i,j,l)$, an element, not in $M$, belonging to level $l$ is loaded. On any step of the last stage, an element of $M$ is loaded. Let us estimate the complexity of the learning graph. Assume for the moment the flow is almost symmetric.

Approximately $n$ arcs are leaving a vertex on each step. Let $(i,j,l)$ be a step of the preparatory phase and assume $l>1$. An element of level $l$ is loaded, and there are $\Omega(r_{l-1})$ $(l-1)$-subtuples in the vertex that can be extended. Hence, $\Omega(r_{l-1})$ arcs leaving the vertex can be used by the flow. This makes the speciality of the step equal to $O(n/r_{l-1})$. This is true for $l=1$ as well, because of the convention $r_0 = n$.

Now turn to the last stage. Let us calculate the speciality of a vertex used by the flow on step $j>1$ of the last stage. Let $V_0$ be the vertices contained in $A_{\ge 1}$ having a valid specification, and $V_M$ be the vertices of $A_{\ge1}\cup M$ that can be used by the flow. Define relation $\varphi$, where $S_0\in V_0$ is in relation with $S_M\in V_M$ if $S_M$ can be obtained from $S_0$ by removing one of its $(j-1)$-subtuples and adding $j-1$ elements from $M$ instead. Each $S_0$ has $\Omega(r_{j-1})$ images and each $S_M$ has $O(n)$ preimages. Hence, $|V_0|/|V_M|=O(n/r_{j-1})$. Because only $O(1)$, out of $\Theta(n)$ arcs leaving a vertex from $V_M$, can be used by the flow, we have the speciality of step $j$ of the last stage equal to $O(n^2/r_{j-1})$. This also is true for $j=1$. All this is summarized in \reftbl{param}.

\begin{table}[htb]
\centering 
\begin{tabular}{|r|ccc|}
\hline
Step & First stage & Preparatory, $(\cdot,\cdot,l)$ & Last stage, $j$-th \\
\hline
Speciality & 1 & $n/r_{l-1}$ & $n^2/r_{j-1}$ \\
Number & $r_1$ & $r_l$ & 1 \\
\hline
\end{tabular}\caption{Parameters (up to a constant factor) of the stages of the learning graph for the $k$-distinctness problem.}
\label{tbl:param}
\end{table}

If we could apply \refcor{symmetric}, we would get the complexity
\[ O\s{r_1 + r_2\sqrt{n/r_1} + r_3\sqrt{n/r_2} + \cdots + r_{k-1}\sqrt{n/r_{k-2}} + n/\sqrt{r_{k-1}}}.\]
Denote $\rho_i = \log_n r_i$ and assume all the addends are equal. Then
\[ \frac12 + \rho_i - \frac{\rho_{i-1}}2  = \frac12 + \rho_{i+1} - \frac{\rho_i}2, \qquad i=1,\dots,k-1\]
where we assume $\rho_0 = 1$ and $\rho_k = 1/2$. It is equivalent to
$ \rho_{i} - \rho_{i+1} = (\rho_{i-1} - \rho_{i})/2$. Hence,
\[
1/2 = \rho_0 - \rho_k = (2^k-1)(\rho_{k-1} - \rho_{k}).
\]
Thus, the optimal choice of $\rho_1$ is $1-2^{k-2}/(2^k-1)$.

We use these calculations to make our choice of $r_i = n^{\rho_i}$. It remains to strictly define the flow, prove it is almost symmetric and the estimates in \reftbl{param} are correct. Before doing so, we combine some estimates on the values of $r_i$'s in the following
\begin{prp}
\label{prp:estimates}
We have $\sqrt{r_1}r_2=o(n)$ and $\sqrt{r_1}=o(r_i)$ for any $i$. Also, any valid specification on stage $i$, has $\Theta(r_j)$ $j$-subtuples for $j<i$.
\end{prp}

\pfstart
The first equation follows from $\rho_1 < 3/4$ and $\rho_2< 5/8$. The second inequality follows from $\rho_i\ge 1/2$ for all $i$'s.

Due to \refline{deadends} of the algorithm, after the first stage, any valid specification has $O(r_1^i/n^{i-1})$ $i$-subtuples. For $i>1$, it is $o(\sqrt{n})=o(r_j)$ for any $j$. Hence, after the first stage there are $\Theta(r_1)$ 1-subtuples, and this number does not substantially change after that. Similarly, if one doesn't take into account the $\pm1$-fluctuations, the number of $j$-subtuples is changed only on stage $j$, when $r_j$ $j$-subtuples are added.
\pfend

\paragraph{Values of the flow} Let us describe how the flow is defined. Fix some stage $i$. A vertex before a step of the form $(i,\cdot,1)$ is called a {\em key vertex}. Consider a key vertex $S$ with type $(b_{t,s})$. The flow from $S$ is distributed evenly among all {\em succeeding} key vertices, where $S'$ is a succeeding key vertex for $S$ iff $S'\setminus S$ is a subset of equal elements having a value different from any element of $S$. The number of succeding key vertices for $S$ is 
\[ N(S) = D_i \s{\sum\nolimits_{t=1}^{k-1} b_{t,1}, \dots, \sum\nolimits_{t=1}^{k-1} b_{t,k-1}} \]
where
\[ D_i(z_1,\dots, z_{k-1}) = \sum_{s=i}^{k-1} \s{\ell_s - z_s}{s\choose i} \]
is the number of possible $i$-subtuples to extend the vertex with, when $z_s$ $s$-tuples have already been used.

More precisely, let $e$ be an arc of step $(i,j,l)$ originating in a non-dead-end vertex $S'$ and loading an element $v$. Then the flow through this arc is defined using the values of the flow through key vertices before step $(i,j,1)$ as follows:
\begin{equation}
\label{eqn:flow_arc}
p_e = \begin{cases}
{s \choose i} {s \choose l}^{-1}\frac{p_{S'\setminus Q}}{lN(S'\setminus Q)},& \parbox{8cm}{$|Q|=l$ and $s\ge i$, where $Q=\{\iota\in S'\cup\{v\} \mid x_\iota = x_v\}$ and $s$ is such that $Q$ is contained in $A_s$;}\\
0,& \mbox{otherwise.}
\end{cases}
\end{equation}
If the first case in~\refeqn{flow_arc} holds, vertex $S'\setminus Q$ is called the key vertex {\em preceeding} arc $e$. Note that it is uniquely defined.

\subsection{Analysis of the flow}
\label{sec:analysis}
\paragraph{Typical vertices} The point of this section is to prove the flow defined in \refsec{flow} is almost symmetric. For this, we should identify the set of typical arcs. Before doing so, we define {\em typical vertices}.

Let $\beta = (b_t)$ be a valid specification of the preparatory phase. Select any $t\in[k-1]$ and let $X_t$ be the collection of all subsets of $A_{\ge 1}$ consisting of $b_t$ $t$-subtuples. In other words, elements of $X_t$ satisfy specification $(0,\dots,0,b_t,0\dots,0)$. Denote $e_{t,s} = \E_{S\in X_t} [b_{t,s}(S)]$, where $b_{t,s}(S) = |S\cap A_s|/t$ is the element of $\tilde\beta(S)$. Denote $\eps_\beta=(e_{t,s})$.

A type $(b_{t,s})$, consistent with $\beta$, is called {\em typical} if it is inside $\cB(\eps_\beta, C\sqrt{r_1})$, where $C$ is a constant to be specified later, i.e., if for all $t$ and $s$ holds $|b_{t,s}-e_{t,s}|\le C\sqrt{r_1}$. A typical vertex is one of a typical type. Let us state some properties of the typical vertices.

\begin{lem}
\label{lem:typical_mean}
Let $(b_{t,s})$ be the type of any typical vertex on the $i$-th stage. Then, for all $t<i$ and $s\ge t$, we have $b_{t,s}=\Omega(r_{t})$.
\end{lem}

\pfstart
Since $\sqrt{r_1} = o(r_t)$, it is enough to show that $e_{t,s} = \Omega(r_t)$. Let $S$ be an element of $X_t$. Arbitrarily order its subtuples: $S=\{s_1,\dots,s_{b_t}\}$. Clearly, the expectation is the same for ordered and unordered lists of subtuples, so let us consider the former.

By linearity of expectation, $e_{t,s} = b_t\pr[s_1\subseteq A_s]$. The number of sequences having $s_1$ in $A_s$ is
$\ell_s{s\choose t}$ times the number of ways to pick the remaining $b_t-1$ $t$-subtuples out of $A_{\ge1}$ where one $s$-tuple cannot be used. By~\refeqn{lk-1} and \reflem{Aprim}, these numbers are equal for different $s\ge t$, up to a constant factor. Hence, the probability is $\Omega(1)$, and since $b_t=\Omega(r_t)$, we have $e_{t,s} = \Omega(r_t)$.
\pfend

\begin{lem}
\label{lem:typical_deviation}
For any valid specification $\beta$ of the preparatory phase and for any $\lambda>C\sqrt{r_1}$,
\begin{equation}
\label{eqn:typical_dev}
\pr[\| \tilde\beta(S) - \eps_\beta\|_\infty>\lambda] < e^{-\Omega(\lambda^2/r_1)},
\end{equation}
where the probability is uniform over all subsets $S$ of $A_{\ge1}$ satisfying $\beta$.
\end{lem}

We derive the lemma from the following two pure technical results
\begin{prp}
\label{prp:deviation}
Let $H$ be the disjoint union of $(H_t)_{t\in[k]}$ where each $H_t$ is a rectangular array of dots, having $\ell_t$ columns and $m_t$ rows. Let $X$ be the set of all $r$-element subsets of $H$ where no subset has more than 1 dot from any column of any $H_t$. Occupy $X$ with the uniform probability distribution and let $h_t\colon X\ni S\mapsto |S\cap H_t|$. Assume $k=O(1)$, $\ell_t = \Theta(n)$ and $r=o(n)$. Then:
\begin{equation}
\label{eqn:deviation}
\pr\left[|h_t - \E[h_t]|>\lambda\right] < e^{-\Omega(\lambda^2/r)},
\end{equation}
for any $\lambda>0$.
\end{prp}

\pfstart
This is a standard application of Azuma's inequality (\refthm{azuma}). Suppose, we sort the elements of each $S\in X$ in any order: $S=\{s_1,\dots,s_r\}$. Clearly, the probability equals for unsorted and for sorted lists. We use both interchangeably in the proof.

Let $D_i$ be the Doob martingale with respect to this sequence. We have to prove that $|D_{i}-D_{i-1}|=O(1)$, i.e., the expectation of $h_t$ does not change much when a new element of the sequence is revealed. To simplify notations, we prove only $|D_1-D_0|=O(1)$, the remaining inequalities being similar.

For the proof, we define two other classes of probability distributions, all being uniform:
\begin{itemize}
\item $Y_i$: $r$-subsets of $H\setminus Q$, where $Q$ is a fixed column of $H_i$;
\item $Z_i$: $(r-1)$-subsets of $H\setminus Q$.
\end{itemize}
For the martingale, it is enough to prove that, for all $i$:
\[
\left| \E[h_t\mid X] - \s{\E[h_t\mid Z_i]+\delta_{i,t}}\right|=O(1),
\]
where $\delta_{i,t}$ is the Kronecker delta. We have
\begin{equation}
\label{eqn:devX}
\E[h_t\mid X] = \pr[Q\cap S\ne\emptyset]\s{\E[h_t\mid Z_i]+\delta_{i,t}} + \pr[Q\cap S = \emptyset]\E[h_t\mid Y_i].
\end{equation}
Hence, it is enough to prove that 
\begin{equation}
\label{eqn:devOutside}
|\E[h_t\mid Z_i] - \E[h_t\mid Y_i]|=O(1).
\end{equation}
Denote $\ell_j' = \ell_j - \delta_{i,j}$; and let 
$K_h$ and $K'_h$ be the number of elements $S$ of $Z_i$ and $Y_i$, respectively, having $h_t(S)=h$. Note that
$K'_h = \gamma_h K_{h-1}$, where $h\in[r]$ and $\gamma_h = ((\ell'_t - h +1) m_t)/h$. 
Thus
\[
\E[h_t\mid Y_i] \le \frac{\sum_{h=1}^{r} h K_h'}{\sum_{h=1}^{r} K_h'} = \frac{\sum_{h=1}^{r} h\gamma_h K_{h-1}}{\sum_{h=1}^{r} \gamma_h K_{h-1}} \le \frac{\sum_{h=1}^{r} h K_{h-1}}{\sum_{h=1}^{r} K_{h-1}} = 1 + \E[h_t\mid Z_i],
\]
where the second inequality holds because $\gamma_h$ monotonely decreases. Thus, by linearity of expectation, $\E[h_t\mid Y_i]\ge \E[h_t\mid Z_i]-k+1$, for all $i$, thus proving~\refeqn{devOutside}. An application of Azuma's inequality finishes the proof of the proposition.
\pfend

\begin{lem}
\label{lem:integrals}
Assume $i,j\in[k-1]$, at least one of them is not 1, and $m=O(1)$ is an integer. Let $\mu$ be a probability distribution on $\R^m$ such that $\mu(\R^m\setminus \cB(\lambda))\le e^{-C_1 \lambda^2/r_i}$ for any $\lambda \ge C_2\sqrt{r_i}$. Assume $w$ is a positive real function, defined on the support of $\mu$, such that $w(x)/w(y) \le e^{C_3 r_j\|x-y\|_\infty/n}$ for any $x,y$. Here $C_1,C_2,C_3$ are some positive constants. Then there exists a constant $C>0$ such that
\[\int_{\R^m\setminus \cB(\lambda)} w(x)\, d\mu(x) = e^{-\Omega(\lambda^2/r_1)}\int_{\R^m} w(x)\,d\mu(x)\]
for any $\lambda \ge C\sqrt{r_1}$.
\end{lem}

\pfstart In the proof, $C$ with a subindex denotes a positive constant that may depend on other $C$'s.

Let $\nu$ be a measure on $]C_2\sqrt{r_i}, +\infty[$ such that $\nu(]\lambda,+\infty[) = \mu(\R^m\setminus \cB(\lambda))$. The worst case, when the mass of $\mu$ is as far from the origin as possible, is when $\nu(]\lambda,+\infty[) = e^{-C_1\lambda^2/r_i}$. In this case, $\nu(t) = g(t)\, dt$ with $g(t)=\frac{2C_1 t}{r_i} e^{-C_1t^2/r_i}$.

There exists a point $y$ in the support of $\mu$ such that $\|y\|_\infty \le C_2\sqrt{r_i}$. Without loss of generality, we may assume $w(y)=1$. Consider
\[
D = \int_{\cB(C_2\sqrt{r_i})} w(x)\,d\mu(x) \ge \mu(\cB(C_2\sqrt{r_i}))\inf_{x\in \cB(C_2\sqrt{r_i})}w(x) \ge (1-e^{-C_1C_2^2}) e^{-2C_2C_3 r_j\sqrt{r_i}/n}.
\]
Then, for any $\lambda\ge C_2\sqrt{r_i}$,
\begin{equation}
\label{eqn:integrals1}
\begin{aligned}
\frac1D \int_{\R^m\setminus \cB(\lambda)} w(x)\,d\mu(x)&\le \frac1D \int_\lambda^{+\infty} e^{C_3 r_j (t+C_2\sqrt{r_i})/n} g(t)\,dt \\ &=
\int_\lambda^{+\infty} \frac{C_4 t}{r_i}\exp\s{C_5 \frac{r_j\sqrt{r_i}}{n} + C_3\frac{r_j t}{n} - C_1 \frac{t^2}{r_i} }dt.
\end{aligned}
\end{equation}
Denote $\tilde t = t/\sqrt{r_1}$. Then the expression in the last exponent can be rewritten as
\[
C_5 \frac{r_j\sqrt{r_i}}{n} + C_3\frac{r_j t}{n} - C_1 \frac{t^2}{r_i} = C_5 \frac{r_j\sqrt{r_i}}{n} + C_3\frac{r_j\sqrt{r_1}}{n} \tilde t - C_1 \frac{r_1}{r_i}{\tilde t}^2 = \frac{r_1}{r_i}\s{C_5 \frac{r_i^{3/2}r_j}{ nr_1 } + C_3\frac{r_i r_j}{ n\sqrt{r_1}} \tilde t - C_1 {\tilde t}^2}.
\]
The coefficients of the last polynomial can be estimated as follows:
\[
\frac{r_i^{3/2}r_j}{ nr_1 } \le \frac{\sqrt{r_1}r_2}{n} = O(1)\quad\mbox{and}\quad \frac{r_i r_j}{ n\sqrt{r_1}} \le \frac{\sqrt{r_1}r_2}{n} = O(1),
\]
by \refprp{estimates}. This means there exist $C_6,C_7>0$ such that, for any $\lambda\ge C_6\sqrt{r_1}$, the right hand side of~\refeqn{integrals1} is at most
\[
\int_\lambda^{+\infty} \frac{C_4 t}{r_i} e^{-C_7t^2/r_i}\,dt = \frac{C_4}{2C_7} e^{-C_7 \lambda^2/r_i}= e^{-\Omega(\lambda^2/r_1)},
\]
if $\lambda \ge C\sqrt{r_1}$ for $C$ large enough.
\pfend

\pfstart[Proof of \reflem{typical_deviation}]
Let $S$ be the random subset. Denote the set of $t$-subtuples of $S$ by $S_t$. We apply \refprp{deviation} to $S_t$ with $k-1$ $H_t$'s given by $m_s = {s\choose t}$ and $r=b_t=O(r_t)$. Thus, if $S_t$ had uniform distribution, Eq.~\refeqn{deviation} would hold, that would imply~\refeqn{typical_dev}, because there are $O(1)$ possible choices of $s$ and $t$.

But in $S$, $S_t$ does not have uniform distribution. Each $S_t$ is assigned weight $w_{S_t}$ that is proportional to the number of subsets of $A_{\ge 1}'$ having specification $(b'_t)$, where $A_{\ge1}'$ has $\ell_j-h_j(S_t)$ $j$-tuples in the notations of \refprp{deviation}, $b'_t=0$ and $b'_j=b_j$ for $j\ne t$.

Take two $S_t$ and $S_t'$, and assume $\|\tilde\beta(S_t)-\tilde\beta(S'_t)\|_\infty \le d$. We apply \reflem{Aprim}. There are two cases. If $t>1$, the lemma implies $w_{S_t}/w_{S_t'}= e^{O(dr_1/n)}$. If $t=1$ then $w_{S_t}/w_{S_t'}= e^{O(dr_2/n)}$. Anyway, either $r$ in~\refeqn{deviation}, or $r$ in the estimation of $w_{S_t}/w_{S_t'}$ is not $r_1$, and, hence, \reflem{integrals} applies, finishing the proof of the lemma.
\pfend

\paragraph{Divergence in the flow} After we have defined typical vertices, we are going to show that almost all flow goes through them. But before we do so, we show get an estimate of the divergence of the flow in the distance of the types.

\begin{lem}
\label{lem:flowTypical}
Suppose two key vertices $S$ and $S'$ of the same specification satisfy $\|\tilde\beta(S)-\tilde\beta(S')\|_\infty\le d$. Then $p_S/p_{S'} = e^{O(dr_2/n)}$.
\end{lem}

\pfstart
Denote $(b_t)=\beta(S)=\beta(S')$, and $b=\sum_t b_t$. Let the original specification of the vertices be $(c_t)$, and $c=\sum_t c_t$.

Fix some order of subtuples in $S$ and $S'$ so that the sizes of the $i$-th subtuple in $S$ and $S'$ are equal for any $i$. Denote this common value by $\nu(i)$. Also, let $\delta_s(i)$ be 1 if the $i$-th subtuple of $S$ is contained in $A_s$, and 0 otherwise. Define $\delta'$ for $S'$ similarly.

Let $\Sigma$ be the set of possible sequences of how the subtuples could have been loaded. I.e., for each element of $\Sigma$, the first $c$ subtuples have specification $(c_t)$, and the remaining $b-c$ subtuples are in a non-decreasing order with respect to their sizes. Moreover, the order of the first $c$ subtuples is irrelevant, i.e., no two distinct elements of $\Sigma$ have their tails of last $b-c$ subtuples equal. In these notations,
\begin{equation}
\label{eqn:flow}
p_S = p_o \sum_{\sigma\in \Sigma} \prod_{j=c+1}^{b} D_{\nu(\sigma j)}\s{\sum\nolimits_{i=1}^{j-1} \delta_1(\sigma i),\dots,\sum\nolimits_{i=1}^{j-1} \delta_{k-1} (\sigma i)}^{-1},
\end{equation}
where $p_o$ and $D$ are defined in \refsec{flow}. A similar expression works for $S'$ as well, if one replaces $\delta$ by $\delta'$.

Since the distance between the types of $S$ and $S'$ is $d$, one can define the order of the subtuples so that $\delta_s(i) = \delta'_s(i)$ for all $s$'s and all, except at most $O(d)$, $i$'s. In this case, for all $\sigma, s$ and $j$:
\[ \left|\sum_{i=1}^{j-1} \delta_s(\sigma i) - \sum_{i=1}^{j-1} \delta'_s(\sigma i) \right|= O(d).\] 
Then the ratio of the $D$'s in \refeqn{flow} is at most $1+O(d/n)$. Since there are $O(r_2)$ multipliers, the ratio of the products in~\refeqn{flow} for the same $\sigma$ is at most
\[\s{1+O\s{\frac dn}}^{O(r_2)} = e^{O(r_2d/n)}.\]
And the same estimate holds for the ratio of sums. \pfend

\paragraph{Finishing the proof} Finally, we are about to prove that the statement of \refcor{symmetric} applies for the flow. Call an arc on preparatory or last stage typical if the preceding key vertex is typical and the flow through the arc is non-zero. We show that conditions of~\refeqn{almost} hold for a fixed value of $x\in f^{-1}(1)$. Then the existence of a strong equivalence between any two positive inputs, as in \refsec{convention}, implies that~\refeqn{almost} holds for all positive inputs $x$ with the values of $\pi(E)$ and $\tau(E)$ independent on $x$.

Note that the factor ${s \choose i} /\left(l{s \choose l}N(S'\setminus Q)\right)$ from~\refeqn{flow_arc} is equal for all arcs from a fixed equivalence class of the preparatory stage, up to a constant factor. Thus, the main concern is about $p_{S'\setminus Q}$, that is flow through a key vertex. The same is true for the last stage as well.

We start with the third condition of~\refeqn{almost}. It is enough to show the flow differs by at most a constant factor for any two typical key vertices of the same specification. The latter follows from the fact the types of typical vertices are at distance $O(\sqrt{r_1})$, and, hence, by \reflem{flowTypical}, the ratio of the flow is $e^{O(r_2\sqrt{r_1}/n)} = O(1)$. 

We continue with the second condition. Again, it is enough to show its analog for key vertices. The latter is a direct consequence of \reflem{integrals} applied to the estimates of Lemmas~\ref{lem:typical_deviation} and~\ref{lem:flowTypical}. The constant $C$ in the definition of the typical vertex is that from the last application of \reflem{integrals}.

Finally, let us calculate the speciality of each step. Because of \refcor{negative}, we may calculate the speciality as if the set of input variables is reduced to $A_{\ge1}\cup M$. Consider a typical arc $e$ of step $(i,j,l)$. Let $S$ be the origin of $e$. Note that $S$ is typical (this is a consequence of~\refeqn{devOutside}). If $l=1$ then we can add any element from an untouched tuple of $A_{\ge i}$. Due to~\refeqn{lk-1}, there are $\Omega(n)$ such elements. 

Now assume $l>1$. By the construction of the flow, the non-zero flow is through the arcs that load the $l$-th element for a subtuple from $A_{\ge i}$. By \reflem{typical_mean}, in $S$, there are $\Omega(r_{l-1})$ $(l-1)$-subtuples from $A_{\ge i}$. In both cases, there are $\Omega(r_{l-1})$ arcs leaving $S$ that are used by the flow. By \reflem{typical_deviation}, an $\Omega(1)$ fraction of all vertices is typical, hence, the speciality of an equivalence class of step $(i,j,l)$ is $O(n/r_{l-1})$.

For the last stage, the same argument as in \refsec{flow} applies, concluded by a fact an $\Omega(1)$ fraction of all vertices before the last stage is typical. 

Thus, the flow is almost symmetric and estimates from~\reftbl{param} are correct. This proves \refthm{main}.

\section{Summary}
An algorithm for $k$-distinctness problem is constructed in the paper, given the prior knowledge of the structure of the input. Is it true, the problem can be solved in the same number of queries without any prior knowledge?

Also, the algorithm in \refsec{previous} can be used for any function such that its 1-certificate complexity is bounded by $k$. For the algorithm in \refsec{main}, it is not clear. So, another (stronger) open problem is as follows. Is it true, any function with 1-certificate complexity bounded by constant can be calculated in $o(n^{3/4})$ quantum queries? If so, this would be a far-reaching generalization of the quantum algorithm in \cite{grafy}.

\subsection*{Acknowledgements}
AB would like to thank Andris Ambainis for useful discussions. AB has been supported by the European Social Fund within the project ``Support for Doctoral Studies at University of Latvia''.

\bibliographystyle{alpha}
\bibliography{span_distinct}

\end{document}